\documentclass[copyright,creativecommons]{eptcs}
\providecommand{\event}{GANDALF 2015} 
\usepackage{breakurl}             
\usepackage[T1]{fontenc}

\usepackage{amssymb}
\usepackage{amsmath,amsthm,yhmath}
\usepackage{thm-restate}

\hypersetup{colorlinks=true,citecolor=blue,linkcolor=red}

\usepackage{nico}
\usepackage{partialinfo}

\begin{document}
\hfuzz=3pt

\title{\ATLsc with partial observation\thanks{This work was partly supported
    by ERC Starting grant EQualIS (FP7-308087) and by FET project Cassting
    (FP7-601148).}}
\author{Fran\c cois~Laroussinie
\institute{LIAFA \\ Univ. Paris Diderot \& CNRS, France}
\and 
 Nicolas~Markey
\institute{LSV\\ ENS Cachan \& CNRS, France}
\and 
 Arnaud~Sangnier
\institute{LIAFA \\ Univ. Paris Diderot \& CNRS, France}
}
\def\authorrunning{Fran\c cois Laroussinie, Nicolas Markey, Arnaud Sangnier}
\def\titlerunning{\ATLsc with partial observation}
\date{\today}

\maketitle

\begin{abstract}
Alternating-time temporal logic with strategy contexts (\ATLsc) is a powerful
formalism for expressing properties of multi-agent systems: it~extends \CTL
with \emph{strategy quantifiers}, offering a convenient way of expressing both
collaboration and antagonism between several agents. Incomplete observation of
the state space is a desirable feature in such a framework, but it quickly
leads to undecidable verification problems. In this paper, we prove that
\emph{uniform} incomplete observation (where all players have the same
observation) preserves decidability of the model checking problem, even for
very expressive logics such as \ATLsc. 

\end{abstract}

\section{Introduction}
Model checking is a powerful technique for automatically checking properties
of computerized systems~\cite{Pnu77,CE82,QS82a}. Model-checking
  algorithms classically take
as input a model of the system under analysis (e.g.~a~finite-state automaton),
and a formal property (expressed e.g. in some temporal logic, such as \LTL
or~\CTL) to be checked; they then automatically and exhaustively verify whether
the set of behaviors of the model satisfies the property.

During the last 15 years, model checking has been extended to handle complex
systems, whose behaviors are the result of the interactions of several
components. Games played on graphs are a convenient model for representing
such interactions, and temporal logics have been hence proposed in
  order to express relevant
properties in such a setting. One of the specification language to
  navigate in the execution trees of multi-agents systems is the temporal
  logic \ATL~\cite{jacm49(5)-AHK} (Alternating-Time Temporal Logic); it is an
  extension of the branching
  temporal logic \CTL which allows to express properties such as the fact that a component
can enforce a certain behavior independently of the actions performed
 by the other components. \ATL has then be enriched in different ways
  to obtain more expressive logics for multi-agent
systems. In particular, \ATLsc (\ATL with strategy contexts)~\cite{BDLM09,LM15} and
Strategy Logic~\cite{CHP07b,MMV10a} are two powerful extensions with similar
properties in terms of expressive power and algorithmic properties.
It~was furthermore proved that those two logics have decidable, but
\ComplexityFont{Tower}-complete model-checking algorithms. 

In the approaches cited above, it~is always assumed that all the players in
the games have perfect observation of the state of the game, and that they
also have perfect recall of the sequence of states that have been visited. In
other words, they can choose an action to perform based on the entire sequence
of states visited before. However, in many applications,
components only have bounded memory, and most often they do not have the
ability to fully observe all the other components of the system. While
considering imperfect recall---the~hypothesis that each player can only
store into a finite memory the history of the seen states seen---can~greatly
simplify verification algorithms (since the number of strategies in the
systems becomes finite), partial observation is known to make \ATL model
checking undecidable~\cite{jacm49(5)-AHK,DT11}. Such results obviously carry
over to more expressive logics like~\ATLsc. Decidability can be regained
by restricting to imperfect-recall strategies~\cite{Sch04a}, or by considering 
hierarchical information~\cite{BMV15} or 
special communication architectures in distributed synthesis~\cite{KV01b,Sch08b}.

In this paper, we consider a restricted case of partial observation, where all
the players have the same information about the state space. We call such a
case \emph{uniform} partial observation. We~prove that under this hypothesis,
model checking concurrent game structure is decidable, even for the powerful
logic~\ATLsc.
In~particular, it~is decidable whether there exists a strategy, based only on
a subset of atomic propositions (assuming that the precise states and the
other propositions are not visible),
to enforce a given property. Note also that the restriction to uniform
observation is not significant when one looks for a strategy of a single
agent~$a$ against all other players, since the semantics we use
requires that
the strategy be winning for any outcome (hence the exact observation of the
opponent players is irrelevant).
%
The decidability proof for model-checking under uniform partial observation is
obtained by adapting a previous approach developed in~\cite{DLM12,LM15}, which
consists in transforming the model-checking problem for \ATLsc* into a
model-checking problem for \QCTL, an extension of \CTL with propositional
quantification. A~similar technique also allows us to prove that when
restricting the strategy quantifiers to range over memoryless strategies, then
the model-checking problem for \ATLsc* with partial observation is again
decidable. We~finally prove that satisfiability checking for~\ATLsc with
partial observation (i.e., deciding whether there exists a game structure with
partial observation satisfying a given formula of~$\ATLsc*$) is undecidable,
even in the case of turn-based games (where satisfiability is decidable under
full observation~\cite{LM13}).



\section{Definitions}
\subsection{Game structures with partial observation}

In this paper, we consider concurrent games with partial observation. They
correspond to classical concurrent game structures~\cite{jacm49(5)-AHK} where, for
each agent, an equivalence relation over the states of the structure defines
sets of states that are observationally equivalent for this player.
Observation equivalence extends to sequences of states in the obvious way.
The~strategies of the agents then have to be compatible with their
observation, in the sense that after two observationally equivalent plays,
a~strategy has to return the same action. In~this preliminary section,
we~formalize this setting.

All along this paper, we consider a set $\AP$ of atomic propositions.
We~recall that a \emph{Kripke structure} over~\AP is a tuple
$\tuple{Q,R,\ell}$ where $Q$ is a countable set of states, $R \subseteq Q
\times Q$ is the transition relation and $\ell\colon Q \to 2^{\AP}$ is a
state-labeling function.

\begin{definition}\label{def-CGSo}
A \emph{concurrent game structure with partial observation}~(\emph{\CGSO{}})
is a tuple $\calC=\tuple{Q, R, \ell, \Agt, \Alac, \Chc, \Edg,(\sim_a)_{a \in
    \Agt}}$ where: $\tuple{Q,R,\ell}$ is a finite-state Kripke structure;
$\Agt=\{a_1,\ldots,a_p\}$ is a finite set of \emph{agents}
(or~\emph{players}); $\Alac=\{m_1,\ldots,m_r\}$~is a finite set of moves
(or~\emph{actions}); $\Chc\colon Q \times\Agt \to
2^\Alac\smallsetminus\{\varnothing\}$ defines the set of available moves of
each agent in each state; $\Edg\colon Q \times \Alac^\Agt \to Q$ is a
transition table associating, with each state~$q$ and each set of moves of the
agents, the resulting state~$q'$, with the requirement that $(q,q')\in R$;
finally, $(\sim_a)_{a \in \Agt}$~assigns to each player an equivalence
relation over~$Q$.
\end{definition}

In the following, we assume w.l.o.g. (w.r.t~existence of specific strategies)
that $\Chc(q,a)=\Alac$ for every $q \in Q$ and $a \in \Agt$: all the actions
are available to all the players at any time. A~move vector is a vector
$\vect{m} \in \Alac^\Agt$; for such a vector and for an agent $a \in \Agt$, we
denote by $\vect{m}[a]$ the move of agent~$a$ in~$\vect{m}$.
A~\newdef{turn-based game structure with partial observation}~(\emph{\TBGSO})
is a \CGSO for which there exists a mapping $\Own\colon Q \to \Agt$ such that
for any~$q\in Q$ and for any two move vectors~$\vect{m}$ and~$\vect{m}'$
if~$\vect{m}[\Own(q)]=\vect{m}'[\Own(q)]$, then $\Edg(q,\vect{m}) =
\Edg(q,\vect{m}')$. 



As we mentioned above, each relation $\sim_a$ with $a \in \Agt$
characterizes the observation of the agent~$a$ in the game
structure: if~$q \sim_a q'$, then agent~$a$
is not able to make a distinction between the states~$q$ and~$q'$;
in~such a case, we~say that $q$ and~$q'$ are \emph{$a$-equivalent}. 
As~an important special case, an observation relation $(\sim_a)_{a \in \Agt}$
is said to be \emph{uniform} when $\mathord{\sim_a}=\mathord{\sim_{a'}}$
for all $a,a'\in\Agt$. 
In that case we might substitute the set $(\sim_a)_{a \in
  \Agt}$ by a unique equivalence relation~$\sim$.

A \emph{finite path} in the \CGSO is a finite non-empty sequence of states
$\rho=q_0 q_1 q_2 \ldots q_k$ such that $(q_i,q_{i+1}) \in R$ for all $i \in
\{0,\ldots,k-1\}$. 
An~\emph{infinite path} (or~run) is an infinite sequence of states
such that each finite prefix is a finite path. We~denote by~$\Path$
(resp.~$\InfPath$) the set of finite (resp.~infinite) paths. Let~$|\rho|$ the
length of the path~$\rho$ (with~$|\rho|=\infty$ if $\rho$~is infinite). For~$0
\leq i < |\rho|$, we~write~$\rho(i)$ to represent the $i+1$-th element of the
path~$\rho$. For~a path~$\rho$,
we write $\first(\rho)$ for its first element~$\rho(0)$ and, when
$|\rho|<\infty$, we~write $\last(\rho)$ for its last element $\rho(|\rho|)$.
For~$0 \leq i < |\rho|$, we~denote by
$\rho_{\leq i}$ the prefix of~$\rho$ until position~$i$, i.e.~the~finite path
$\rho(0) \rho(1) \ldots \rho(i)$. We~extend the equivalence relation $\sim_a$
for $a \in \Agt$ to paths as follows : two~paths $\rho$ and~$\rho'$ are
$a$-equivalent (written~$\rho \sim_a \rho'$) if, and only~if, $|\rho|=|\rho'|$
and $\rho(i)\sim_a \rho'(i)$ for every $0\leq i < |\rho|$.

Given a \CGSO~$\calC$ and one of its states~$q_0$, we~write $\calT_\calC(q_0)$
for the \emph{execution tree} of~$\calC$ from~$q_0$: formally,
$\calT_\calC(q_0)$~is the pair~$\tuple{T,\ell}$ where $T$ is the set of all
finite paths (called \emph{nodes} in the context of trees) in~$\calC$ with
first state~$q_0$, and $\ell$ labels each node~$\rho$ with the labeling
of~$\last(\rho)$ in~$\calC$. It~will be convenient in the sequel to see
execution trees as infinite-state Kripke structures. To~alleviate notations,
we~still write $\calT_\calC(q_0)$ for the Kripke structure $\tuple{T, R,
  \ell}$ where $\tuple{T,\ell}$ is the tree defined above, and $R\subseteq
T\times T$ is the transition relation such that $(\rho,\rho')\in R$ whenever
$\rho$~is the prefix of~$\rho'$ of length~$|\rho'|-1$.

A \emph{strategy} for agent $a \in \Agt$ is a function $f_a \colon \Path \to
\Alac$; it~associates with any finite path a move to be played by agent~$a$
after this path. A~strategy~$f_a$ for agent~$a \in \Agt$ is said to be
\emph{memoryless} whenever for any two finite paths $\rho$ and~$\rho'$ such
that $\rho(|\rho|-1)=\rho'(|\rho'|-1)$, it~holds $f_a(\rho)=f_a(\rho')$. Hence
the decision of a memoryless strategy depends only on the current control
state; for this reason, we may simply give such a strategy as a fonction
$f_a\colon Q \to \Alac$. A~strategy~$f_A$ for a coalition of agents $A
\subseteq \Agt$ is a set of strategy $\{f_a\}_{a\in A}$ assigning a
strategy~$f_a$ to each agent $a \in A$ (note that a strategy for agent~$a$ is
equivalent to a strategy for coalition~$\{a\}$). Given a strategy
$f_A=\{f_a\}_{a\in A}$ for coalition~$A$, we~say that a path~$\rho$
respects~$f_A$ from a finite path~$\pi$ if, and only~if, for all $0 \leq i <
|\pi|$, we have $\rho(i)=\pi(i)$ and for all $|\pi| \leq i < |\rho|-1$, we
have that $\rho(i+1)= \Edg(\rho(i),\vect{m})$ where $\vect{m}$ is a move
vector satisfying $\vect{m}(a)=f_a(\rho_{\leq _i})$ for all $a \in A$. Given a
finite path $\pi$, we~denote by $\Out(\pi,f_A)$ the set of infinite
paths~$\rho'$ such 
that $\rho$ respects the strategy~$f_A$ from~$\pi$. Given a strategy
$g_A=\{g_a\}_{a \in A}$ for a coalition~$A$ and a strategy $f_B=\{f_b\}_{b \in
  B}$ for a coalition~$B$, we denote by $g_A\compo f_B$, the~strategy
$\{h_c\}_{c \in A \cup B}$ for coalition $A \cup B$ such that $h_c=g_c$ for
all $c \in A$ and $h_c=f_c$ for all $c \in B \setminus A$. Finally given a
$f_B=\{f_b\}_{b \in B}$ for a coalition~$B$ and a set of agents $A \subseteq
\Agt$, we~denote by $(f_B)_{| A}$ (resp.~$(f_B)_{\smallsetminus A}$) the
strategy $\{f_b\}_{b\in B \cap A}$ (resp.~$\{f_b\}_{b\in B \setminus A}$) for
coalition $B \cap A$ (resp.~$B \setminus A$).

Partial observation comes into the play by restricting the space of allowed
strategies: 
in our setting, we only consider strategies that are \emph{compatible} with
the observation in the game, which means that after any two $a$-equivalent
finite paths~$\rho$ and~$\rho'$, the~strategies for agent~$a$ have to take
the same decisions (i.e.~$f_a(\rho)=f_a(\rho')$).
We~could equivalently define a compatible strategy for~$a$ as a function from
the quotient set $\Path_{/\sim_a}$ to~$\Alac$, such that if $[\rho]$ is the
equivalence class of~$\rho$ with respect to~$\sim_a$, then $f_a([\rho])$
gives the move to play for~$a$ from any history equivalent to~$\rho$.
A~strategy $f_A=\{f_a\}_{a \in A}$ for coalition $A$ is \emph{compatible} if
$f_a$ is compatible for all $a \in A$. 
We~write $\Strat(A)$ to denote the unrestricted strategies for coalition~$A$,
$\StratS(A)$~for the set of compatible strategies, and $\StratSZ(A)$ is the
set of compatible memoryless strategies for~$A$.

\subsection{\ATL with strategy contexts}
We will be interested in the logic \ATLsc, which extends the alternating-time
temporal logic of~\cite{jacm49(5)-AHK} with strategy contexts. We assume a fixed set
of atomic propositions~$\AP$ and a fixed set of agents~$\Agt$.
\begin{definition}\label{def-ATLSes}
  The formulas of~\ATLsc* are defined by the following grammar:
  \begin{xalignat*}1 
    \ATLsc* \ni \phis,\psis \coloncolonequals & P \mid
    \non\phis \mid \phis\ou\psis
    \mid \Diams[A] \phip \mid \Diamsco[A] \phip \mid 
     \Relax[A] \phip \mid  
     \Relaxco[A] \phip 
    \\
    \phip,\psip \coloncolonequals & \phis \mid \non\phip \mid \phip\ou\psip
    \mid \X\phip\mid \phip\Until\psip
\end{xalignat*}
where $P$~ranges  over~$\AP$ and $A$~ranges over~$2^\Agt$.
\end{definition}


We interpret \ATLsc* formulas over \CGSOs, within a context (\ie,
a~preselected strategy for a coalition): \newdef{state formulas} of the
form~$\phis$ in the grammar above are evaluated over states, while
\newdef{path formulas} of the form~$\phip$ are evaluated along infinite paths.
In~order to have a uniform definition, we~evaluate all formulas at a given
position along a path.

In~\ATLsc*, contrary to the case of classical~\ATL~\cite{jacm49(5)-AHK},
when strategy quantifiers assign new strategies to some players, the other
players keep on playing their previously-assigned strategies. This~is what
``strategy contexts'' refer~to. Informally, formula~$\Diams[A]\phip$ holds at
position~$n$ along~$\rho$ under the context~$F$ if it~is possible to
extend~$F$ with a \emph{strategy} for the coalition~$A$ such that the outcomes
of the resulting strategy after~$\rho_{\leq n}$ all satisfy~$\phip$.
In~Section~\ref{sec-mem}, we~also consider the strategy quantifiers
$\Diamsm[A]$, which quantifies only on \emph{memoryless} strategies.
Finally strategies can
be dropped from the context using the operator~$\Relax$. 
Notice that in this paper, all strategy quantifiers are restricted to
compatible strategies.

We~now define the semantics formally. Let~$\calC$ be a \CGSO with agent
set~$\Agt$, $\rho$~be an infinite path of~$\calC$, and $n\in\bbN$. 
Let~$B\subseteq\Agt$ be a coalition, and
$f_B\in\StratS(B)$. That a (state or~path) formula~$\phi$ holds at a
position~$n$ along~$\rho$ in~$\calC$ under strategy context~$f_B$, denoted
${\calC, \rho, n \models_{f_B}\phi}$, is defined inductively as follows
(omitting atomic propositions and Boolean operators):
\begin{xalignat*}{3}
  \calC, \rho, n   &\models_{f_B}  \Relax[A] \phis &&\quad\mbox{ iff }\quad 
    &&\calC, \rho, n \models_{(f_B)_{\smallsetminus A}} \phis \\
\calC, \rho, n &\models_{f_B}  \Relaxco[A] \phis &&\quad\mbox{ iff }\quad 
    &&\calC, \rho, n \models_{(f_B)_{| A}} \phis \\
\noalign{\pagebreak[1]}
  \calC, \rho, n &\models_{f_B}  \Diams[A] \phip && \quad\mbox{ iff }\quad   
     &&\exists g_A\in\StratS(A). \
     \forall \rho'\in\Out(\rho_{\leq n},g_A\compo f_B).\ \calC,
     \rho',n \models_{g_A\compo f_B} \phip \\
\noalign{\pagebreak[1]}
  \calC, \rho, n &\models_{f_B}  \Diamsco[A] \phip && \quad\mbox{ iff }\quad   
  &&\calC, \rho, n \models_{f_B}  \Diams[\Agt\setminus A] \phip \\
%
%
%
\calC, \rho,n &\models_{f_B}  \X \phip &&\quad\mbox{ iff }\quad 
&& \calC, \rho,n+1 \models_{f_B} \phip \\ 
\calC, \rho,n  &\models_{f_B}  \phip \Until \psip && \quad\mbox{ iff }\quad
 &&\exists l\geq 0.\ 
  \calC, \rho, n+l \models_{f_B} \psip  \mbox{ and }
  \forall 0\leq m < l.\  \calC, \rho, n+m \models_{f_B} \phip 
\end{xalignat*}


Finally, we~write $\calC, q_0\models \phis$ when
$\calC, \rho , 0\models_{f_\emptyset}\phis$ (with empty context) for
all path $\rho$ such that $\rho(0)=q_0$. 
Notice that this does not depend on the choice of~$\rho$.
%
The usual shorthands such as $\F$ and~$\G$ are defined as for \CTL*. It~will
also be convenient to use the constructs $\Boxs[A] \phip$ as a shorthand for
$\non\Diams[A]\non\phip$, and $\Diams[A] \phis$ as a shorthand for $\Diams[A]
\bot \Until \phis$.
The \CTL* universal path quantifiers $\Ex$ and~$\All$ can be expressed
in~\ATLsc*. For instance, 
$\All \phip$ is equivalent to~$\Relaxco[\emptyset] \Diams[\emptyset] \phip$.
Finally the fragment \ATLsc of \ATLsc* is defined as usual, by restricting the
set of path formulas to 
\[
  \phip,\psip   \coloncolonequals  \non\phip \mid \X\phis \mid  \phis\Until\psis. 
\]
%


\begin{remark}
The strategy quantifiers for complement coalitions (namely
$\Diamsco[A]$ and $\Relaxco[A]$) are only useful in case the set $\Agt$ is not
known in advance, as it is the case when dealing with satisfiability or with
expressiveness questions.
\end{remark}

\begin{remark}
As opposed to Strategy Logic~\cite{CHP07b,MMV10a}, the (existential) strategy
quantifiers in \ATLsc include an implicit (universal) quantification over the
outcomes of the strategies of the unassigned players. This is indeed a
quantification over all the outcomes, and not over the outcomes that would
result only from compatible strategies.
\end{remark}

\begin{example}
Consider the CGSO $\calC$ in Figure~\ref{fig-cgs}. It is a turn-based CGS with
two players (Player $a_1$ plays in circle node, and Player $a_2$ plays in box
node). The set $\Alac$ is $\{1,2,3\}$. The partial observation for Player
$a_1$ is then given by $\sim_{a_1}$ is $\equiv_{P}$ (\ie\ two states are
equivalent if, and only~if, the truth value of $P$ is the same in both
states).
Now consider the formula $\phi=\Diams[a_1] \F f$. To see that $\phi$
holds for true in $q_0$ with the standard semantics (where the strategy
quantifiers are not restricted to compatible strategies), it is sufficient to
consider the (memoryless) strategy for $a_1$ consisting in choosing the
move~$m_1$ from~$q_2$ and the move~$m_2$ (or~$m_3$) from~$q_3$. Note that this
strategy is not compatible (after $q_0q_2$ and $q_0q_3$, agent~$a_1$ should choose the
same move), thus we need to consider another one to show that $\phi$ is
satisfied with the partial-observation semantics; for example, the strategy
consisting in choosing~$m_1$ after~$q_0q_2$ and~$q_0q_3$, $m_2$~after~$q_0q_1q_3$
and~$m_1$ after~$q_0q_3q_2$. If~furthermore we look for a memoryless compatible strategy
for~$a_1$, we~can use the strategy assigning the move~$m_1$ to~$q_2$ and~$q_3$,
hence the formula $\Diamsm[a_1] \F f$ also holds true in~$q_0$. Therefore,
formula~$\phi$ holds true at~$q_0$. On~the other hand, formula
$\Diamsm[a_1] \X \X \X f$ does
not hold true in~$q_0$, but formula $\Diams[a_1] \X \X \X f$, where we drop
the memory constraint, does hold true.
\end{example}

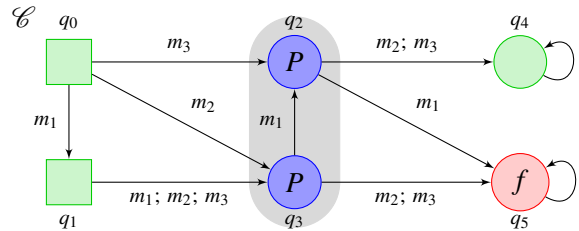
\begin{figure}[!ht]
\centering
\begin{tikzpicture}
\tikzstyle{every node}=[rectangle]
\draw (1,2.8) node[style=carre,style=vert] (q0) {}
  node[above=3mm] {$\scriptstyle q_0$};
\draw (1,1.2) node[style=carre,style=vert] (q1) {}
  node[below=3mm] {$\scriptstyle q_1$};

\draw[rounded corners=6mm,fill=black,opacity=.15] (3.4,2) |- (4.6,3.4) |- (3.4,.6) -- (3.4,2);
\draw (4,2.8) node[style=rond,style=bleu] (q2) {$P$}
  node[above=3mm] {$\scriptstyle q_2$};
\draw (4,1.2) node[style=rond,style=bleu] (q3) {$P$}
  node[below=3mm] {$\scriptstyle q_3$};
\draw (7,2.8) node[style=rond,style=vert] (q4) {}
  node[above=3mm] {$\scriptstyle q_4$};
\draw (7,1.2) node[style=rond,style=rouge] (q5) {$f$}
  node[below=3mm] {$\scriptstyle q_5$};

\tikzstyle{every node}=[rectangle]
\foreach \nod / \angle / \lab in
   {q4/0/,q5/0/}
    {\draw[arrows=-latex'] (\nod.\angle-30)
        .. controls +(\angle-30:6mm) and +(\angle+30:6mm)
        .. (\nod.\angle+30) node[pos=.5,right=2pt,above=3pt]
      {\hbox to 0pt{$\scriptstyle\listof{\lab}$\hss}}; }
\foreach \noda / \nodb / \lab / \dec in
    {
      q0/q1/m_1/left,q0/q2/m_3/above,q0/q3/m_2/above right,q1/q3/m_1;\ m_2;\
      m_3/below,q3/q2/m_1/left,q2/q4/m_2;\ m_3/above,q2/q5/m_1/above right,q3/q5/m_2;\ m_3/below}
    {\draw[arrows=-latex'] (\noda) -- (\nodb) node[midway,\dec]
      {$\scriptstyle {\lab}$};}
\path (q0) -- +(-.6,.6) node {$\mathcal C$};
\end{tikzpicture}
\caption{The turn-based \CGS $\calC$}
\label{fig-cgs}
\end{figure}

In the sequel we will consider the \emph{model-checking problem} of  \ATLsc*
over \CGSOs which takes as input a \CGSO $\calC$, a control state $q_0$
and an \ATLsc* formula $\phi$ and which answers yes if $\calC,
q_0\models \phi$ and no otherwise. We will also consider the
\emph{satisfiability problem} which given a formula of  \ATLsc*
determines whether there exists a \CGSO $\calC$ and a control state
$q_0$ such that  $\calC,q_0\models \phi$. We recall that when
considering concurrent game structures (with full observation), \ATLsc*
model-checking is decidable, but its satisfiability problem is undecidable
(except when restricting to turn-based games)~\cite{LM15}.

\subsection{From concurrent games to turn-based games}
Following our definitions, turn-based game structures can be
seen as special cases of concurrent game strcutres, where in each location
only one player may have several non-equivalent moves.
\begin{figure}[t]
\centering
\begin{tikzpicture}
\begin{scope}
\draw (0,0) node[rond,vert] (a) {$q_0$};
\draw (-2,-2.6) node[rond,vert] (b1) {$q_1$};
\draw (0,-2.6) node[rond,vert] (b2) {$q_2$};
\draw (2,-2.6) node[rond,vert] (b3) {$q_3$};
\everymath{\scriptstyle}
\path (a) edge[-latex'] node[left] {$\tuple{m_1,m_1}$} (b1);
\path (a) edge[-latex',bend left] node[pos=.4,left] {$\tuple{m_2,m_1}$}
  node[pos=.6,left] {$\tuple{m_1,m_2}$} (b2);
\path (a) edge[-latex'] node[right] {$\tuple{m_2,m_2}$} (b3);
\end{scope}
\begin{scope}[xshift=7cm]
\draw (0,0) node[rond,vert] (a) {$q_0$};
\draw (-2,-2.6) node[rond,vert] (b1) {$q_1$};
\draw (0,-2.6) node[rond,vert] (b2) {$q_2$};
\draw (2,-2.6) node[rond,vert] (b3) {$q_3$};
\fill[black,opacity=.15,rounded corners=3mm] (0,-.8) -| (1.5,-1.8) -|
  (-1.5,-.8) -- (0,-.8);
\draw (-1,-1.3) node[carre,jaune] (i1) {\phantom{$q_2$}};
\draw (1,-1.3) node[carre,jaune] (i2) {\phantom{$q_3$}};
\everymath{\scriptstyle}
\path (a) edge[-latex'] node[left] {$m_1$} (i1);
\path (a) edge[-latex'] node[right] {$m_2$} (i2);
\path (i1) edge[-latex'] node[left] {$m_1$} (b1);
\path (i1) edge[-latex'] node[left] {$m_2$} (b2);
\path (i2) edge[-latex'] node[right] {$m_1$} (b2);
\path (i2) edge[-latex'] node[right] {$m_2$} (b3);
\end{scope}
\end{tikzpicture}
\caption{Turning a \CGSO into an equivalent \TBGSO}
\label{fig-cgs2tb}
\end{figure}
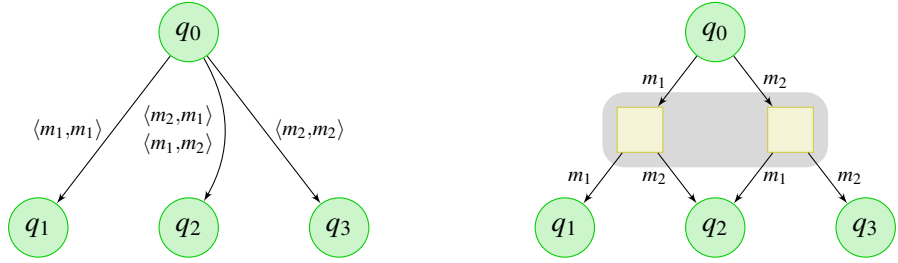
In this section, we prove that any partial-observation \CGSO can be turned
into an \emph{equivalent} partial-observation \TBGSO (where \emph{equivalent}
will be made precise later). While all players play
at the same time in a \CGSO, they play one after the other (in
any predefined order) in the correspondig \TBGSO, but the intermediary states
are made 
undistinguishable to all players, so that no player can gain information
from playing after another one. 
Figure~\ref{fig-cgs2tb} schematically represents this transformation in the
case of two players. 
Obviously, since we add intermediary states,
we also have to modify the \ATLsc* formula to be checked, by making the
intermediary states ``invisible''; this is a classical construction in temporal
logics. In~the end, we~prove the following result:
\begin{restatable}{theorem}{cgstb}
\label{thm-cgs2tb}
  For any \CGSO~$\calC=\tuple{Q,R,\ell,\Agt,\Alac,\Chc,\Edg,(\sim_a)_{a \in
      \Agt}}$, we can compute 
  a
  \TBGSO~$\calT=\tuple{Q',R',\ell',\Agt,\Alac,\Chc',\Edg',(\approx_a)_{a \in
      \Agt}}$ with $Q\subseteq Q'$, for which there is a logspace
  translation~$\phi\in\ATLsc* \mapsto \widetilde\phi\in \ATLsc*$ such that for
  any~$\phi\in\ATLsc*$ and any state~$q$ of~$\calC$, it~holds:
\[
\calC,q\models_\emptyset \phi \quad\Leftrightarrow\quad
  \calT,q\models_\emptyset \widetilde\phi.
\]
Furthermore, if $\calC$ is uniform, then so is~ $\calT$.
\end{restatable}

\subsection{\QCTL* in a nutshell}
As we explain in the sequel, under some restrictions (uniformity or
restriction of the considered strategies), the model-checking problem of
\ATLsc* over \CGSOs can be reduced to the model-checking problem of the
temporal logic Quantified \CTL* over Kripke Structures. \QCTL*~extends the
classical branching time temporal logic~\CTL* with \emph{atomic propositions
  quantifiers} $\exists P.\ \phis$ (and its dual $\forall P.\ \phis$), where $P$
is an atomic proposition in~$\AP$:
\begin{xalignat*}1 
    \QCTL* \ni \phis,\psis \coloncolonequals & P \mid
    \non\phis \mid \phis\ou\psis
    \mid \Ex \phip \mid \exists P.\ \phis
    \\
    \phip,\psip \coloncolonequals & \phis \mid \non\phip \mid \phip\ou\psip
    \mid \X\phip\mid \phip\Until\psip
\end{xalignat*}
where $P$~ranges  over~$\AP$. We~briefly review \QCTL* here, and refer
to~\cite{LM14} for more details and examples.

\QCTL*~formulas are then interpreted over a (classical) Kripke structure
$\calS=\tuple{Q,R,\ell}$. Informally $\exists P.\ \phis$ is used to specify
the existence of a valuation for~$P$ over~$\calS$ such that $\phis$ is
satisfied.
%
We~consider two different semantics of $\exists P. \phi$:  
the \emph{structure semantics} where one looks
for a valuation of the structure~$\calS$, and the \emph{tree semantics} where
one considers the valuation of its execution tree.
%
%
%
%

For $X\subseteq \AP$, two Kripke
structures $\calS=\tuple{Q,R,\ell}$ and $\calS'=\tuple{Q',R',\ell'}$ are
\emph{$X$-equivalent} (denoted by $\calS \equiv_X \calS'$) whenever $Q=Q'$,
$R=R'$, and $\ell \equiv_X \ell'$ (\ie,~${\ell(q)\cap X=\ell'(q)\cap X}$ for
any~$q\in Q$). 
The \newdef{structure semantics}
of 
\QCTL* (whose satisfaction relation is denoted by~$\models_{\mathtt{st}}$) is
derived from the 
semantics of \CTL* by adding the following rule:
\[
\calS, \rho, n\models_{\mathtt{st}} \exists P.\ \phis \quad\mbox{ iff }\quad
  \exists \calS'.\ \calS' \equiv_{\AP\smallsetminus\{P\}} \calS \text{ and }
  \calS',\rho,n\models_{\mathtt{st}} \phis.
\]
In~other terms, $\exists P.\ \phis$ means that it is possible to
(re)label the Kripke structure with~$P$ in order to make $\phis$
hold. 
As~for \CGSO, we~write $\calS, q_0 \models_{\mathtt{st}}
  \phi_s$ whenever $\calS, \rho,0 \models_{\mathtt{st}}
  \phi_s$ for all path~$\rho$ such that $\rho(0)=q_0$. 
The~\emph{tree semantics} (whose satisfaction relation is denoted
by~$\models_{\mathtt{tr}}$) is 
defined as:
$\calS, q_0 \models_{\mathtt{tr}} \phi_s$ if, and only~if,
$\calT_{\calS}(q_0),q_0 \models_{\mathtt{st}} \phi_s$ (where
$\calT_{\calS}(q_0)$ is the execution tree of~$\calS$ from~$q_0$, seen as an
infinite-state Kripke structure).

\begin{example}
\label{rem-qctl-path}
Consider formula $\EX \phi \et \forall P.\Big(\EX (\phi \et P) \impl \AX (P
\impl \phi)\Big)$, which we write $EX_1 \phi$ in the sequel. This
formula states the existence of a unique immediate successor
satisfying~$\phi$. We will reuse it later in our construction.
\end{example}


\begin{theorem}[\cite{LM14}]\label{thm-qctl}
\begin{enumerate}
\item For the structure semantics, the model-checking problem of
  \QCTL* is \PSPACE-complete and the satisfiability problem is
  undecidable.
\item  For the tree semantics, the model-checking and satisfiability problems of
  \QCTL* are decidable, and \ComplexityFont{Tower}-complete.
\end{enumerate}
\end{theorem}


\section{Model checking uniform  \CGSOs}
\label{sec3}

As we have already mentioned, it is well known that the model-checking of \ATL
over \CGSO is undecidable~\cite{jacm49(5)-AHK,DT11}. Since \ATL is a fragment
of \ATLsc*, this undecidability result holds also for \ATLsc*. We~prove in
this section that decidability can be regained by restricting to
uniform partial observation, i.e., 
by assuming that  the observation is the same for all the agents. The~idea
consists in reducing the model-checking problem of \ATLsc* over uniform \CGSOs
to the model-checking problem for \QCTL* with tree semantics over a complete
Kripke structure representing the possible observations.


In the sequel, we consider a uniform \CGSO $\calC=\tuple{Q,
  R, \ell, \Agt, \Alac, \Chc, \Edg,(\sim_a)_{a \in \Agt}}$ with the finite set
  of states $Q=\{q_0,\ldots,q_s\}$, the finite set of moves
  $\Alac=\{m_1,\ldots,m_r\}$ and $\mathord{\sim_a}=\mathord{\sim}$ for all
  $a \in \Agt$. For each $q \in Q$, we~denote by $[q]$ the equivalence class
  of~$q$ with respect to~$\sim$, \ie $[q] \subseteq 2^Q$ is such that $q'
\in [q]$ if, and only~if, ${q \sim q'}$. We~also consider an \ATLsc*
formula $\phi_s$ whose $\Diams$-depth (i.e.~the~maximal number of nested
modalities~$\Diams$) is~$\lambda$ (we~assume here that modalities~$\Diamsco$
and $\Relaxco$ have been previously removed from~$\phis$, as the set of agents
is known).
The idea behind our proof is to consider a Kripke structure~$\calS_{\calC}$ as
a complete graph whose set of nodes is the set of equivalence classes~$[q]$.
We~then use the power of~\QCTL* (with the tree semantics) to build paths and
strategies in this Kripke Structure. The fact that, in the Kripke structure,
the paths go through equivalent classes of states (and not through states),
guarantees that the corresponding strategies are compatible with the partial
observation.

Before we build the Kripke Structure, we need to introduce some sets
of fresh atomic propositions (not appearing in~$\calC$ or
in~$\phis$) which we will use to encode paths and strategies in the
Kripke structure. Let $\APQ \eqdef \{\qsf_i^\ell \mid  1 \leq i \leq
s \ \mbox{and}\  0\leq \kappa \leq \lambda  \}$,
$\APM{a} \eqdef \{\msf_1^a,\ldots,\msf_r^a\}$ for every $a \in \Agt$, and
$\APMt \eqdef 
\bigcup_{a\in\Agt}\APM{a}$ and $\APS=\{ \ssf_i \mid  1 \leq i
\leq s \}$ be those sets of fresh atomic propositions. We~assume that
$\APQ$, $\APMt$ and $\APS$ are subsets of~$\AP$. 
Intuitively, propositions in~$\APQ$
will be used to fix a path of~$\calC$ at each level~$\kappa$ of quantification
(a~path at level~$\kappa$ being a sequence of states~$\qsf_i^\kappa$);
propositions in~$\APM{a}$ will be used to label the choices
corresponding to a strategy of agent~$a$.

The Kripke structure~$\calS_{\calC}$ associated with~$\calC$ is then defined as
$\tuple{Q_\calC,R_\calC,\ell_\calC}$ where ${Q_\calC = \{ [q] \mid
q \in Q\}}$, ${R_\calC = Q_\calC \times Q_\calC}$ and $\ell_\calC([q])
= \{ \ssf_i \mid q_i \in [q] \}$. 
%
Compatible strategies can then be encoded as functions labeling the execution
trees of this Kripke structure (from the state where we evaluate the formula):
a~strategy~$f_a$ for some agent~$a$ is represented as a function
$\widetilde{f_a}\colon Q_\calC^+ \to \APM{a} $ that labels the execution
trees~$\calT_{\calS_{\calC}}(q)$ of~$\calS_{\calC}$ from the current state~$q$ with
proposition in~$\APM{a}$. However note that $\calS_{\calC}$ being a complete
graph, paths in this structure might not correspond to paths in the
associated \CGSO~$\calC$; we will use atomic propositions~$\qsf_{i}$ to
identify concrete paths of~$\calC$.

Given a coalition $B$ in~$\Agt$, an integer $0\leq\kappa \leq
\lambda$, and a subformula of $\phis$ at  $\Diams$-depth $\kappa$, 
we~define a \QCTL* formula $\widetilde{\phi}^{B,\kappa}$ inductively as follows:
\begin{xalignat*}2
\widetilde{P}^{B,\kappa}  &\eqdef  \OU_{\{ i \mid  P\in \ell(q_i)\}} \qsf^\kappa_i   &
\widetilde{\Relax[A] \phis}^{B,\kappa}  &\eqdef  \widetilde{\phis\vphantom{\psi}}^{B\smallsetminus A,\kappa} \\[1mm]
\noalign{\pagebreak[1]}
\widetilde{\phi \et \psi}^{B,\kappa}
&\eqdef  \widetilde{\vphantom{\psi}\phi}^{B,\kappa}  \et 
\widetilde{\vphantom{\psi}\psi}^{B,\kappa}  &
\widetilde{\non \phi}^{B,\kappa}
&\eqdef    \non \widetilde{\vphantom{\psi}\phi}^{B,\kappa}  
\\[3mm]
%
%
\widetilde{\phip \Until \psip}^{B,\kappa} &\eqdef  \widetilde{\vphantom{\psi}\phip}^{B,\kappa}  \Until
\widetilde{\vphantom{\psi}\psip}^{B,\kappa}  &
\widetilde{\X \phip}^{B,\kappa}  &\eqdef    \X \widetilde{\vphantom{\psi}\phip}^{B,\kappa} 
\end{xalignat*}
For a formula of the shape\footnote{For the sake of readability, we
  restrict to one-player coalitions here; the construction easily
  extends to the the general case with a coalition (including the empty coalition).} $\Diams[a] \phip$, the construction is 
as follows:
\begin{multline*}
  \widetilde{\Diams[a] \phip}^{B,\kappa} \eqdef
    \exists \msf_1^{a} \ldots \msf_\kappa^{a}.  \Big[ \Phistrat(\{a\}) \et \forall
 \qsf_{1}^{\kappa+1}\ldots\qsf_{s}^{\kappa+1}. \Big( \Phipath(\kappa+1) \et  \\
   \OU_{1\leq i\leq s}  (\qsf_{i}^\kappa \et \qsf_i^{\kappa+1})   \et \Phiout(\kappa+1,B\cup \{a\}) \Big) \impl  \All \Big(\G \OU_{1\leq i \leq s} \qsf^{\kappa+1}_i  \;\impl\; 
    \widetilde{\vphantom{\phi}\phip}^{B\cup\{a\},\kappa+1}\Big)\Big] 
\end{multline*}
%
%
with
\begin{xalignat*}1
\Phistrat(A) &\eqdef  \ET_{a\in A} \AG \Big(\OU_{1\leq j \leq r}
  (\msf_j^{a} \et \ET_{j'\not=j}  \non\msf_{j'}^a)\Big) 
  \\[1mm]
\Phipath(\kappa+1) &  \eqdef \EG \Big(\OU_{1\leq i \leq
  s} \big(\qsf_i^{\kappa+1} \et \ssf_i \et \ET_{i'\not=
  i} \non \qsf_{i'}^{\kappa+1}\big) \et \EX_1 (\OU_{1\leq i \leq
  s} \qsf_i^{\kappa+1})\Big) \\[1mm] 
\Phiout(\kappa+1,B) &\eqdef \EG \Big( \OU_{1\leq i \leq s} 
  \OU_{\bigl\{(\vect{m},q_j) \ \bigm|\
  \begin{array}{@{}l@{}} 
    \scriptstyle \vect{m} \in \Alac^\Agt \wedge\\[-2mm]
    \scriptstyle q_j=\Edg(q_i,\vect{m})
  \end{array}\bigr\}} \qsf_i^{\kappa+1} \et 
  \widetilde{\vect{m}}^B \et \EX \qsf^{\kappa+1}_j \Big)  
\end{xalignat*}
\noindent
where $\widetilde{\vect{m}}^B$ is a shorthand for the formula
  ${\ET_{\{(b,m_k) \mid b \in B \wedge \vect{m}[b]=m_k\}}
  m^b_k}$. We point out the fact that $\Phiout(\kappa+1,B)$ is based on the
  transition table~$\Edg$ of~$\calC$; consequently, its size is in
  $O(|Q|^2\cdot|\Alac|^{|\Agt|})$ (\ie,~in~$O(|Q|\cdot|\Edg|)$).

The intuition behind formula $\widetilde{\Diams[a] \phip}^{B,\kappa}$ is as
  follows: the~formula first ``selects'' a strategy for agent~$a$ under the
  form of a labeling of the execution tree of~$\calS_\calC$ with~$\msf_i^{a}$;
  it~then~uses subformula~$\Phistrat(\{a\})$ to ensure that this labeling
  correctly encodes a strategy for~$a$; finally, it~checks that all the
  outcomes of the selected strategy satisfy~$\phip$. The~latter task is
  achieved by considering all the labelings of the structure
  with~$\qsf^{\kappa+1}_{-}$, and, for the labelings that correspond to
  one outcome of the selected strategy, by checking the formula~$\phip$. 
  Ensuring that a labeling corresponds to an outcome is done as follows:
\begin{enumerate}
\item it corresponds to an infinite branch in the
  execution tree of~$\calS_\calC$, and each node labeled
  by~$\qsf^{\kappa+1}_i$ should correspond to a node~$[q_i]$ (labeled
  by~$\ssf_i$) in~$\calS_\calC$ (both~points are ensured by
  formula\footnote{See Example \ref{rem-qctl-path} for the definition of
  $\EX_1$} ${\Phipath(\kappa+1)}$);
\item at the present position, one of the propositions~$\qsf^{\kappa+1}_{i}$
  has to match with one of the state~$\qsf^{\kappa}_{i}$ of the previous level
  (this ensures that the path labeled by~$\qsf^{\kappa+1}_{-}$ starts from the
  ``current state'' considered in the game);
\item the branch obtained by this labeling effectively follows the
  choices dictated by the labels~$\msf_{-}^{b}$ encoding the strategies for
  $b \in B 
  \cup \{a\}$; this is checked by the formula  $\Phiout(\kappa+1,B
  \cup \{a\})$. 
\end{enumerate}
Finally, the formula checks that the corresponding path satisfies the formula
$\widetilde{\vphantom{\phi}\phip}^{B\cup\{a\},\kappa+1}$.

\medskip
The correctness of the reduction is stated in the following theorem:
\begin{theorem}
\label{th-mc-uii}
Let~$\phi_s$ be an $\ATLsc*$ state-formula, $\calC$
be a uniform \CGSO, and $q_\alpha$ be a state of~$\calC$.
Then:
\[
\calC,q_\alpha \sat_{} \phis \quad\mbox{if, and only~if,}\quad 
\calS_\calC,[q_\alpha] \models_{\mathsf{tr}}\exists \qsf_\alpha^0. \Big( \widetilde{\phis}^{\emptyset,0}\et \qsf^0_{\alpha} \Big)
\]
\end{theorem} 

\begin{proof}
%
%
First we point out the fact that none of the formula used in the reduction
checks that the considered strategies are compatible, but in fact this is
guaranteed because we evaluate the formula over the Kripke
structure~$\calS_\calC$, and two equivalent paths in~$\calC$ with respect to
$\sim$ are necessarily matched to a unique path in the execution tree of
$\calS_\calC$.

We now prove that our reduction is correct. Let
$\calT_{\calS_\calC}([q_\alpha])=\tuple{T_\calC,R'_\calC,\ell'_\calC}$ be the
execution tree of the Kripke structure~$\calS_\calC$ from the state $[q_\alpha]$. Note that nodes
in~$\calT_{\calS_\calC}([q_\alpha])$ are thus finite paths of the form $[q_0][q_1]
\ldots[q_k]$ (with $q_0=q_\alpha$). For~a finite path~$\pi$ in~$\calC$
such that $\pi(0)=q_\alpha$, we~denote by~$[\pi]$ the
path in~$\calS_\calC$ such that $|\pi|=|[\pi]|$ and for all $0 \leq i <
|\pi|$, we have $[\pi](i)=[\pi(i)]$. We~also denote by~$\pathf(\pi)$ the path
in~$\calT_{\calS_\calC}([q_\alpha])$ such that $|\pathf(\pi)|=|\pi|$ and for all $0 \leq i
< |\pi|$, $\pathf(\pi)(i)=[\pi_{\leq i}]$. We~now consider the tree
$\calT'=\tuple{T_\calC,R'_\calC,\ell'}$ where $\ell'$~extends the
function~$\ell'_\calS$ with propositions of type~$\msf_j^a$
and~$\qsf^\kappa_i$.
This extension is used to encode a strategy context and to describe a
path of~$\calC$. Given a compatible strategy $f_B=\{f_b\}_{b \in B}$ for the
coalition~$B=\{b_1,\ldots,b_h\}$, an~infinite  path~$\rho$ in~$\calC$
such $\rho(0)=q_\alpha$
and a position~$n \geq 0$, we say that $\ell'$~is:
\begin{itemize}
\item an $f_B$-labeling whenever, for every node $\gamma \in T_\calC$
  with $\gamma=[\pi]$ for some finite path~$\pi$ in~$\calC$, for any $b
  \in B$, for any $1\leq j \leq r$, we have $\msf^b_j \in
  \ell'(\gamma)$ if, and only~if,  $f_b(\pi)=m_j$;
\item a $(\kappa,\rho)$-labeling if the following two conditions are
  verified:
  \begin{enumerate}
   \item for all  $j \geq 0$,   
    it~holds $\qsf^\kappa_i \in
    \ell'([\rho_{\leq j}])$ if, and only~if, $\rho(j)=q_i$ 
  \item if $q_i \in \ell'(\pi)$ for a node  $\pi \in T_\calC$, 
    then there exists $j \geq 0$   
    such that $\pi=[\rho_{\leq j}]$. 
  \end{enumerate}
  In other words, the
  propositions $\qsf^\kappa_1,\ldots,\qsf^\kappa_s$ label a unique
  branch in the tree, that can be matched with the path~$\rho$.
\end{itemize}

We say that for $0 \leq \kappa \leq \lambda$, a~subformula~$\phi$ of~$\phis$ 
occurs at   $\Diams$-depth~$\kappa$, if the
number of $\Diams$-quantifier above~$\phi$ in tree representing
formula~$\phis$ is~$\kappa$. 

\medskip
\begin{proposition}
\label{prop-mc-uii}
Let $\phi$ be some subformula of~$\phis$ which occurs at
$\Diams$-depth $\kappa$ with $0 \leq \kappa \leq \lambda$, $\rho$~be a
run of~$\calC$ such that $\rho(0)=q_\alpha$, $n$~be a natural number,
$f_B$~be a compatible strategy for coalition~$B$. 
Let $\calT'=\tuple{T_\calC,R'_\calC,\ell'}$ be the tree defined as
above. 
If~$\ell'$ is an $f_B$-labeling
and a $(\kappa,\rho)$-labeling, then we have:
\[
\calC,\rho,n \sat_{f_B} \phi \quad\mbox{if, and only~if,}\quad 
\calT',\pathf(\rho),n \models_\mathsf{st} \widetilde{\phi}^{B,\kappa}. 
\] 
\end{proposition}


\begin{proof}
The proof is done by structural induction over~$\phi$. In~the following, we
let $q_{\alpha'}=\rho(n)$.  
\begin{itemize}
\item 
case $\phi\eqdef P$: we~have $\rho(n)\sat_{f_B} P$ if, and only~if, $P\in
\ell(q_{\alpha'})$. As~$\ell'$ is a $(\kappa,\rho)$-labeling, we~know that
$\calT',\pathf(\rho),n \sat \qsf^\kappa_\alpha$. By~definition
of~$\widetilde{P}^{B,\kappa}$, the implication follows. Conversely assume
$\calT',\pathf(\rho),n \sat \widetilde{P}^{B,\kappa}$, we~know that $\rho(n)$
has to be labeled by~$P$, because $\ell'$~is a $(\kappa,B)$-labeling.
\item 
case $\phi\eqdef \phip \Until \psip$: if $\calC,\rho,n \sat_{f_B}
\phip\Until\psip$, then there exists $i\geq n$ s.t.\ $\calC,\rho,i
\sat_{f_B} \psip$ and for any $n \leq j < i$, we~have $\calC,\rho,j
\sat_{f_B} \phip$. By~i.h., we~get $\calT',\pathf(\rho),i \sat
\widetilde{\psip}^{B,\kappa}$ and, for any~$j$,  
$\calT',\pathf(\rho),j \sat \widetilde{\phip}^{B,\kappa}$; from this we obtains
$\calT',\pathf(\rho),n \sat \widetilde{\phi}^{B,\kappa}$. The converse is
similar. 
\item 
case $\phi\eqdef \Diams[a] \phip$:  
Assume $\calC,\rho,n\sat_{f_B} \Diams[a] \phip$. Then there exists a
$\sim$-compatible strategy~$f_a$ s.t.~for any $\rho' \in \Out(\rho_{\leq
  n},f_a\compo f_C)$, we have $\rho' \sat_{f_a\compo f_B} \phip$. From this
strategy~$f_a$, we~deduce a valuation for propositions
$\msf^a_1,\ldots,\msf^a_r$ extending~$\ell'$ over~$T$ (because $f_a$ is
$\sim$-compatible), and satisfying $\Phistrat(\{a\})$. Now extend $\ell'$ with
a valuation for $\qsf^{\kappa+1}_1,\ldots,\qsf^{\kappa+1}_s$ following the run
$\rho'$ (i.e.~for every state $\rho'(i)=q_\beta$, the corresponding
node~$[\rho'](i)$ is labeled by $\qsf^{\kappa+1}_\beta$, and only
$[\rho']$-nodes are labeled by these propositions). Let $\ell''$ be this new
valuation for~$\calT'$.Then clearly we have:
\begin{itemize}
\item $\calT',\pathf(\rho),n \sat \Phipath(\kappa+1)$, since $\qsf$ propositions
  label a path; 
\item $\calT',\pathf(\rho),n \sat \qsf^\kappa_{\alpha'} \et
  \qsf^{\kappa+1}_{\alpha'}$, because the current position belongs to the
  runs~$\rho$ and~$\rho'$, and the new run~$\rho'$ is issued
  from the current  position;  
\item $\calT',\pathf(\rho),n \sat \Phiout(\kappa+1,C\cup\{a\})$, meaning that
  the labeling of $\qsf^{\kappa+1}_{-}$ propositions follows the "correct"
  path~$\rho'$ from $\Out(\rho_{\leq n},f_a\compo f_B)$;
\item finally, $\calT',\pathf(\rho'),n  \sat
  \widetilde{\phip}^{B\cup\{a\},\kappa+1}$  by~i.h. 
\end{itemize}
Therefore we have: $\calT',\pathf(\rho),n  \sat \widetilde{\phi}^{B,\kappa}$.

We now prove the converse implication. Assume $\calT',\pathf(\rho),n \sat
\widetilde{\phi}^{B,\kappa}$. From the existence of a labeling for
$\msf^a_1,\ldots,\msf^a_r$ satisfying~$\Phistrat(\{a\})$, we~deduce a
$\sim$-compatible strategy~$f_a$ in~$\calC$ for every finite runs issued
from~$\rho(\leq n)$
Now consider a valuation for $\qsf^{\kappa+1}_1,\ldots,\qsf^{\kappa+1}_s$.
Either it makes the left-hand side subformula of the implication to be false,
and there is no consequence, or this subformula is true and in this case, the
valuation describes a run~$\rho'$ in~$\calC$ issued from~$\rho(n)$ and
belonging in $\Out(\rho_{\leq n},f_a\compo f_B)$; this run has to satisfy
$\widetilde {\phip}^{B\cup\{a\},\kappa+1}$, and by~i.h.\ we get that $\calC,\rho',n
\sat_{f_a\compo f_B} \phip$.\qed\qed
\end{itemize}
\let\qed\relax
\end{proof}
\let\qed\relax

\end{proof}

\begin{corollary}
The model-checking problem of  \ATLsc*
over uniform \CGSOs is decidable (and \textsf{Tower}-complete).
\end{corollary}

\begin{remark}
Our algorithm can be used to decide whether one player (with partial
observation of the system) has a compatible strategy to win against all the
other players, whatever is the observation of the other players (since the
implicit quantification in strategy quantifiers ranges over all the outcomes).
As~a consequence, when considering the fragment of \ATL in which strategy
quantifiers always involve at most one-player coalitions, the model-checking
problem is decidable for possibly non-uniform partial observation. It~can be
checked that indeed the undecidability proof of~\cite{DT11} involves a
strategy quantification over a coalition of two players with different
observation.
\end{remark}

\section{Restriction to memoryless strategies}
\label{sec-mem}
In this section, we show that another way to obtain decidability for the
model-checking problem of \ATLsc* over game structure with (possibly
non-uniform) partial observation is by restricting the set of considered
strategies to memoryless strategies. We~denote by
$\ATLsc:$, the temporal logic obtained from \ATLsc* by replacing respectively
the quantifiers over strategies $\Diams[A]$ by $\Diamsm[A]$ and $\Diamsco[A]$ by
$\Diamsmco[A]$. Given a \CGSO $\calC$ a path $\rho$ of~$\calC$, and $n\in\bbN$
such that $n < |\rho|$ and $f_B\in\StratSZ(B)$, the semantics of these two
operators is given by:
\begin{xalignat*}{2}
\calC, \rho, n \models_{f_B}  \Diamsm[A] \phip & \quad\mbox{ iff }\quad   
      \exists g_A\in\StratSZ(A). \
      \forall \rho'\in\Out(\rho_{\leq n},g_A\compo f_B).\ \calC, \rho',n
      \models_{g_A\compo f_B} \phip \\
 \calC, \rho, n \models_{f_B}  \Diamsmco[A] \phip & \quad\mbox{ iff }\quad   
      \exists g_{\bar{A}}\in\StratSZ(\Agt \setminus A). \
      \forall \rho'\in\Out(\rho_{\leq n},g_{\bar{A}}\compo f_B).\ \calC, \rho',n \models_{g_{\bar{A}}\compo f_B} \phip 
\end{xalignat*}

In~\cite{LM15}, a reduction from model-checking $\ATLsc0$ and $\ATLsc:$
to \QCTL* model-checking problem over concurrent game structure with perfect
observation is given, the main idea is to use the structure semantics instead
of the tree semantics. In~this framework the complexity is simpler, since both
problems are \PSPACE-complete: the~number of memoryless strategies for one
player being bounded (by~$|\Alac|^{|Q|}$), we~can easily enumerate all of
them, and store each strategy within polynomial space. When considering
partial observation, we~can use the same approach, we need to also ensure that
the chosen strategies are compatible, but this can easily be achieved when
considering memoryless strategies, since one only needs to check that a
strategy~$f_a$ proposes for two equivalent states (w.r.t.~$\sim_a$) the same
move.

Let $\calC=\tuple{Q, R, \ell,
  \Agt, \Alac, \Chc, \Edg,(\sim_a)_{a \in \Agt}}$ be a \CGSO with finite set
of states
$Q=\{q_0,\ldots,q_s\}$, finite set of moves
$\Alac=\{m_1,\ldots,m_r\}$ and the finite set of agents
$\Agt=\{a_1,\ldots,a_p\}$. For each $a \in \Agt$, we suppose that the
number of equivalence classes for $\sim_a$ is $n_a$ and we use the notation
$E^a_i$ with $1 \leq i \leq n_a$ to represent the $i$-th equivalence
class of $\sim_a$. We~denote by~$\calK_\calC$ the Kripke
structure $\tuple{Q,R,\ell_\calK}$ underlying~$\calC$ based on the
same states and the same transitions, and where the labeling function
extends~$\ell$ by adding labels from the sets $\{\psf_q \mid q \in Q\}
\cup \{\psf^{\sim_a}_i \mid a \in \Agt \wedge 1 \leq i \leq n_a \}$
in such a way that the following two conditions hold:
\begin{itemize}
\item $\psf_q \in \ell_\calK(q')$ if, and only~if, $q'=q$,
\item for all $a \in \Agt$, $\psf^{\sim_a}_i \in \ell_\calK(q)$ if, and
only~if, $q \in E^a_i$.
\end{itemize}


Below we show how to translate a formula $\phis$ of $\ATLsc:$ into a
formula of $\widehat{\phis}$ of \QCTL* such that we
will have $\calC, q_0\models
 \phis$ for a state
$q_0$ of $\calC$ iff $\calK_\calC,q_0\models_{\mathsf{st}}
\widehat{\phis}$. Note that at the opposite of the previous section,
we consider here the
\emph{structure semantics}~\cite{LM14}  when performing the model-checking
of $\calK_\calC$: this means that instead of ranging over labelings of the
execution tree, propositional quantification ranges over
labelings of the Kripke structure. This is due to the fact that the
compatible memoryless  strategies we take here into account are
functions mapping equivalent states (and not anymore path of
equivalent states) to the same  move. Given a coalition $B$ in~$\Agt$  and $\phi$ a subformula of $\phis$, we  define a \QCTL* formula $\widehat{\phi}^{B}$
inductively as follows:
\begin{xalignat*}2
\widehat{P}^{B}  &\eqdef  P  &
\widehat{\Relax[A] \phis}^{B}  &\eqdef  \widehat{\phis\vphantom{\psi}}^{B\smallsetminus A} 
\\
\widehat{\phi \et \psi}^{B} &\eqdef  \widehat{\vphantom{\psi}\phi}^{B}  \et
\widehat{\vphantom{\psi}\psi}^{B}  &
\widehat{\non \phi}^{B}  &\eqdef    \non \widehat{\vphantom{\psi}\phi}^{B} 
\\
\widehat{\phip \Until \psip}^{B} &\eqdef  \widehat{\vphantom{\psi}\phip}^{B}  \Until
\widehat{\vphantom{\psi}\psip}^{B}  &
\widehat{\X \phip}^{B}  &\eqdef    \X \widehat{\vphantom{\psi}\phip}^{B} 
\end{xalignat*}

\noindent
For a formula of the shape $\Diamsm[a] \phip$,
we let\footnote{As in the previous case, the extension to   $\Diamsm[A]$ is
straightforward.}:  
\[
\widehat{\Diamsm[a] \phip}^B \eqdef \exists \msf_1^{r} \ldots
\msf_\kappa^{a}. 
\Big( \Phistratm(\{a\}) \et \All \Bigl[( \Phiout'(\{a\}\cup B)) \;\impl\;
\widehat{\vphantom{\phi}\phip}^{B\cup \{a\}}\Bigr]\Big)
\]
where 
\[
\Phistratm(A) \eqdef \ET_{a\in A} \biggl[ \AG \Big(\OU_{1\leq j \leq r} (\msf_j^{a} \et \ET_{j'\not=j}  \non\msf_{j'}^a)\Big) \et 
\ET_{1\leq c \leq n_a} \ET_{1\leq j \leq  r} \EF (\psf^{\sim_a}_c \et \msf^a_j) \impl \AG (\psf^{\sim_a}_c \impl \msf^a_j) \biggr]
\]
and
\[
\Phiout'(B) = \G\Big(
  \ET_{q\in Q} \
  \ET_{\bigl\{(\vect{m},q') \big|
    \begin{array}{@{}l@{}}
      \scriptstyle\vect{m} \in \Alac^\Agt   \wedge \\[-1mm] 
      \scriptstyle q'=\Edg(q,\vect{m})
      \end{array}
      \bigr\}} 
  \left[(\psf_q \et \widehat{\vect{m}}^B) \thn  \X\psf_{q'}
  \right]
\Big).
\]
%
%
where $\widehat{\vect{m}}^B$ is a shorthand for the formula ${\ET_{\{(b,m_k)
  \mid b \in B \wedge \vect{m}[b]=m_k\}} m^b_k}$. The last
Formula~$\Phiout'(B)$ characterizes the outcomes of the strategies in
use for some coalition~$B$. We point out that the \QCTL*
formula~$\widehat{\phis}^\emptyset$ has size $O(|\Phi|\cdot |Q|\cdot (|\Agt|\cdot|\Alac|^2 + |Q|\cdot|\Edg|))$. 

A memoryless and compatible strategy $f_B$ can easily be encoded by
labeling states of  $\calK_\calC$ with the moves of the coalition $B$ thanks to the propositions $\msf^a_1,\ldots,\msf^a_r$ for $a \in B$. 
The $\Phistratm(B)$ is here to ensure both that the labels $\msf^a_\_$
correspond to a strategy for coalition $B$ and to check that such a
strategy is memoryless and compatible: for all \emph{reachable} states
labeled with some move, the formula verifies that all the equivalent
states are labeled by the same move. Note that this property ranges
over reachable states: we could have unreachable states that are
labeled in an improper way but this is not a problem as such state
will never be reached from the current position (and consequently they
cannot impact on the truth value of the formula). 

 We will now consider the Kripke structure $\calK'_\calC=\tuple{Q,R,\ell'}$  where $\ell'$ extends the function
$\ell_\calK$  with propositions of type $\msf_j^a$. Given a compatible
memoryless strategy $f_B=\{f_b\}_{b \in B}$ for the
coalition~$B$, we say that $\ell'$ is an
$f_B$-labeling if   for every node $q  \in Q$, for any $b
  \in B$ and for any $1\leq j \leq r$, we have $\msf^b_j \in
  \ell'(q)$ iff  $f_b(q)=m_j$ (we recall that a memoryless strategy
  $f_b$ can be seen as a function from $Q$ to $\Alac$).

\begin{proposition}
\label{prop-mc-uii0}
Let $\rho$~be a
run of~$\calC$, $n$~be a natural number,
$f_B$ be a compatible strategy in $\StratSZ(B)$ for coalition $B$. If $\ell'$ is an $f_B$-labeling, then we have:

\[
\calC,\rho,n \sat_{f_B} \phis \quad\mbox{iff}\quad \calK'_\calC,\rho,n \sat \widehat{\phis}^{B} 
\]
%
\end{proposition}

\begin{proof}
The proof is done by structural induction over $\Phi$.
\begin{itemize}
\item $\Phi \eqdef \phip \Until \psip$: Assume $\calC,\rho,n \sat_{f_B} \phip \Until \psip$. Therefore there exists $i\geq n$ s.t.\ $\calC,\rho,i \sat_{f_B}  \psip$ and for any $n \leq j < i$, we have $\calC,\rho,j \sat_{f_B} \phip$. For every position between $n$ and $i$, the induction hypothesis can be applied and we deduce $\calK'_\calC,\rho,n \sat \widehat{\phip}^{B} \Until \widehat{\psip}^{B}$.
The converse is done similarly. 
\item $\Phi \eqdef \Diams[a] \phip$: $(1)\impl (2)$. Assume $\calC,\rho,n \sat_{f_B} \Diams[a] \phip$. There exists a memoryless compatible strategy $f_a$ s.t.\ for any $\rho' \in \Out(\rho_{\leq n},f_a\compo f_B)$, we have $\calC,\rho',n \models_{f_a\compo f_B} \phip$. Thus we can find a labeling of propositions $\msf^a_1,\ldots,\msf^a_r$  for $\calK'_\calC$ to represent $f_a$. This labeling completes the existing one for strategy context $F_B$ and the formula $\Phistratm(\{a\}\cup B)$ is then satisfied on $\calK'_\calC,\rho,n$. And any $\calK'_\calC$-run satisfying $\Phiout'(\{a\}\cup B)$ belongs to the set of outcomes generated by the strategy context $f_a\compo f_B$, and then satisfies $\widehat{\phip}^{\{a\}\cup B}$ by induction hypothesis. \\
\noindent
$(2)\impl (1)$.  Now  assume  $\calK'_\calC,\rho,n \models \widehat{\phi}^{B}$.
Therefore there exists a labeling for   $\msf^a_1,\ldots,\msf^a_r$ such
that $\Phistratm(\{a\})$ holds true. Such a labeling defines a memoryless
and  compatible strategy for the reachable states (from $\rho,n$).   And
finally  every run satisfying $\Phiout'(\{a\}\cup B)$ has to satisfy
$\widehat{\phip}^{\{a\}\cup B}$. By induction hypothesis we get the stated
result.\qed
\end{itemize}
\let\qed\relax
\end{proof}

Thus we have: 

\begin{theorem}\label{thm-atlsco-mc-ii}
The model-checking problem of $\ATLsc:$ is \PSPACE-complete. 
\end{theorem}

\begin{proof}
\QCTL* model checking can be solved in polynomial space for the structure
semantics. This immediately gives a \PSPACE algorithm for our problem.
Conversely, \ATLsc0 model checking (with perfect observation)
is \PSPACE-hard~\cite{BDLM09}. This extends immediately to the
partial-observation setting. 
\end{proof}

\begin{remark}
Quantification over memoryless strategies could also be achieved using the
tree semantics, following the presentation of Section~\ref{sec3}. To~do~so,
it~suffices to label each state with its name (hence adding a few extra atomic
propositions) and to require that the labeling of the execution tree with
strategies satisfies that whenever some state~$s$ is labeled with some
move~$m_i$, then all the occurrences of~$s$ must be labeled with the same
mode~$m_i$. 

While this has little interest for model checking $\ATLsc:$ in terms of
complexity, it~shows that model checking remains decidable for the logic
involving quantification over both memoryful, bounded-memory and memoryless
strategies.
\end{remark}

\section{Satisfiability and partial observation}

In this section, we consider satisfiability checking: given a formula
$\Phi$, we look for a game structure~$\calC$, an equivalence~$\sim$ over
states, and a control state~$q_0$, such that $\calC,q_0 \models \Phi$. This
problem is undecidable for \ATLsc (and~\ATLsc*) in the classical setting of
perfect-observation games.

First note that considering partial observation makes the problem 
different: there exists formulas that are satisfiable under partial
observation, and not satisfiable for full observation.
Consider formula $\Phi \eqdef \AX (\Diams[a_1] \X
f) \et \non \Diams[a_1] \X\X f$. Clearly $\Phi$ is satisfiable when
considering games with partial observation: for~example, one~can consider the
turn-based structure in Figure~\ref{fig-ex-sat}, where $a_1$ plays in circle
nodes and $a_2$ plays in box nodes, and where $\sim$~is~$\equiv_{P}$: from
$q_0$, there is no $\sim$-strategy for $a_1$ ensuring the property $\X \X f$
(because from the strategy should play the same move after~$q_0q_1$
and~$q_0q_2$), but the subformula $\Diams[a_1] \X f$ holds for true in~$q_1$
and in~$q_2$.

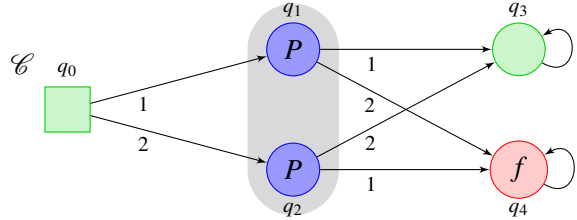
\begin{figure}[!ht]
\centering
\begin{tikzpicture}
\tikzstyle{every node}=[rectangle]
\draw (1,2) node[style=carre,style=vert] (q0) {}
  node[above=3mm] {$\scriptstyle q_0$};
\draw[rounded corners=6mm,fill=black,opacity=.15] (3.4,2) |- (4.6,3.4) |-
  (3.4,.6) -- (3.4,2);
\draw (4,2.8) node[style=rond,style=bleu] (q1) {$P$}
  node[above=3mm] {$\scriptstyle q_1$};
\draw (4,1.2) node[style=rond,style=bleu] (q2) {$P$}
  node[below=3mm] {$\scriptstyle q_2$};
\draw (7,2.8) node[style=rond,style=vert] (q3) {}
  node[above=3mm] {$\scriptstyle q_3$};
\draw (7,1.2) node[style=rond,style=rouge] (q4) {$f$}
  node[below=3mm] {$\scriptstyle q_4$};

\tikzstyle{every node}=[rectangle]
\foreach \nod / \angle / \lab in
   {q3/0/,q4/0/}
    {\draw[arrows=-latex'] (\nod.\angle-30)
        .. controls +(\angle-30:6mm) and +(\angle+30:6mm)
        .. (\nod.\angle+30) node[pos=.4,right=2pt,below=1pt]
      {\hbox to 0pt{$\scriptstyle\listof{\lab}$\hss}}; }
\foreach \noda / \nodb / \lab / \dec in
    {
      q0/q1/1/+0pt,q0/q2/2/+3pt,q1/q3/1/+3pt,q2/q4/1/+3pt,q1/q4/2/+3pt,q2/q3/2/+3pt}
    {\draw[arrows=-latex'] (\noda) -- (\nodb) node[pos=.3,below=-1pt]
      {$\scriptstyle {\lab}$};}
\path (q0) -- +(-.6,.6) node {$\mathcal C$};
\end{tikzpicture}
\caption{The turn-based \CGS $\calC$ s.t.\ $q_0\models\Phi$}
\label{fig-ex-sat}
\end{figure}

But formula $\Phi$ is not satisfiable in the classical case: assume that $\AX
(\Diams[a_1] \X f)$ is satisfied in some state~$q$, then from every
$q$-successor $q'$, there is a strategy $f_{q'}$ for $a_1$ to ensure $\X f$.
Therefore the strategy for~$a_1$ consisting, from~$q$, to choose an arbitrary
move, and then to choose the strategy $f_{q'}$ for every possible
successor~$q'$ ensures the property~$\X\X f$.

From a decidability point of view, considering partial observation
does not make satisfiability problems to be simpler: in fact this
problem remains undecidable for \ATLsc. Furthermore, it is even worse
than in the classical setting: while the turn-based satisfiability is decidable when perfect information is assumed, it is undecidable when one considers the partial observation case.

\begin{theorem}
Satisfiability problems for \ATLsc*. \ATLsc,  $\ATLsc:$ and \ATLsc0 (with partial
observation) are undecidable even restricted to turn-based structures. 
\end{theorem}

\begin{proof}
We can mostly reuse the proof of Troquard and Walther~\cite{jelia2012-TW} (we have slighty modified in \cite{LM15}). The key idea of their proof is to reduce $S5^n$ satisfiability to \ATLsc satisfiability. Given an $S5^n$ formula $\Phi$, one can build an \ATLsc formula $\Diams[\overline{\emptyset}] \wideparen{\Phi}$ such that $\Phi$ is satisfiable iff $\Diams[\overline{\emptyset}]\wideparen{\Phi}$ is satisfiable. Without considering the details of the proof, one can just note that $\wideparen{\Phi}$ uses Boolean operators and strategies quantifiers $\Diams[a_i]$ (for $1\leq i \leq n$) and formulas of the form $\X P$: there is no $\Until$ modality and there is no nesting of $\X$. Therefore with such a formula, every strategies quantifier  is interpreted in a unique state $w$ and we only consider the moves done in this state $w$: adding an  equivalence  $\sim$ over states or considering memoryless strategies do not change the semantics of $\wideparen{\Phi}$.

Assume that $\Phi$ is satisfiable. From the proof of~\cite{jelia2012-TW}, one
can build a game structure $\calC$ satisfying
$\Diams[\overline{\emptyset}] \wideparen{\Phi}$. Moreover one can use the
reduction of Theorem~\ref{thm-cgs2tb} to get a tun-based games with partial
information.

Now assume that $\Diams[\overline{\emptyset}] \wideparen{\Phi}$ is
satisfiable. Therefore there exists a game structure~$\calC$ with an
equivalence~$\sim$ such that $\calC,
q \models \Diams[\overline{\emptyset}] \wideparen{\Phi}$.  
This structure is based on a set of agents $\{1,\ldots,k\}$ with $k\leq n$. 
And there is a strategy~$F$ for~$\Agt$ such that $\calC,
q \models_F \wideparen{\Phi}$ and $\wideparen{\Phi}$ may only   modify the
choices for players $a_1,\ldots,a_n$.  
If~$k>n$, we can replace the players $n+1$,\ldots, $k$ by their first move selected by   $F$ from $q$. Its gives a game structure based on $\Agt=\{a_1,\ldots,a_n\}$. And we can use the same construction of a corresponding $S5^n$ model for $\Phi$ as it is done in~\cite{LM15}, and as explained before, considering partial information does not change the construction since every strategy is applied from the same unique initial state of the game. 
\end{proof}


\section{Conclusion}
In this paper, we have proved that the model-checking of \ATLsc over games
with \emph{uniform} partial observation is decidable. It~would be interesting
to study whether our uniformity requirement on the observation of each player
in the game could be relaxed in order to be able to analyze richer models. One
possible direction is to look at games with hierarchical
information~\cite{BMV15}, but we currently could not find a way of extending
our algorithm to non-uniform observation.


\bibliographystyle{myalpha}
\bibliography{bibexport}

\end{document}